\documentclass[%
 superscriptaddress,
 onecolumn,
 notitlepage,
 showpacs,
 showkeys,
 nofootinbib,
  amsmath,amssymb,
  aps,
 pra,
  longbibliography,
  floatfix,
 ]{revtex4-1}

\usepackage{amssymb}
\usepackage{mathtools}
\usepackage{amsthm}

\newtheorem{Theorem}{Theorem}
\newtheorem{Lemma}[Theorem]{Lemma}

\def\ket#1{\mathinner{|{#1}\rangle}}
\def\braket#1{\mathinner{\langle{#1}\rangle}}

{\catcode`\|=\active 
  \gdef\Braket#1{\left<\mathcode`\|"8000\let|\BraVert {#1}\right>}}
\def\BraVert{\egroup\,\mid@vertical\,\bgroup}

\DeclareMathOperator{\Tr}{Tr}
\DeclareMathOperator{\Span}{Span}
\DeclareMathOperator{\sgn}{sgn}
\newcommand{\T}{\mathsf{T}}
\newcommand{\vect}[1]{\mathbf{#1}}
\newcommand{\vprod}[2]{\vect{#1}\cdot\vect{#2}}
\usepackage{bm}

\usepackage[utf8]{inputenc}

\usepackage{hyperref}

\begin{document}

\title{Tight state-independent uncertainty relations for qubits}

\author{Alastair A. Abbott}
\email{alastair.abbott@neel.cnrs.fr}
\affiliation{Institut N\'{e}el, CNRS and Universit\'{e} Grenoble Alpes, 38042 Grenoble Cedex 9, France}

\author{Pierre-Louis Alzieu}
\affiliation{Institut N\'{e}el, CNRS and Universit\'{e} Grenoble Alpes, 38042 Grenoble Cedex 9, France}
\affiliation{Ecole Normale Sup\'{e}rieure de Lyon, 69342 Lyon, France}

\author{Michael J. W. Hall}
\affiliation{Centre for Quantum Computation and Communication Technology (Australian Research Council),
Centre for Quantum Dynamics, Griffith University, Brisbane, QLD 4111, Australia}

\author{Cyril Branciard}
\email{cyril.branciard@neel.cnrs.fr}
\affiliation{Institut N\'{e}el, CNRS and Universit\'{e} Grenoble Alpes, 38042 Grenoble Cedex 9, France}

\date{\today}

\begin{abstract}
The well-known Robertson--Schr\"{o}dinger uncertainty relations have state-dependent lower bounds which are trivial for certain states.
We present a general approach to deriving tight state-independent uncertainty relations for qubit measurements that completely characterise the obtainable uncertainty values.
This approach can give such relations for any number of observables, and we do so explicitly for arbitrary pairs and triples of qubit measurements.
We show how these relations can be transformed into equivalent tight entropic uncertainty relations.
More generally, they can be expressed in terms of any measure of uncertainty that can be written as a function of the expectation value of the observable for a given state.
\end{abstract}


	
\maketitle


\section{Introduction}

One of the most fundamental features of quantum mechanics is the fact that it is impossible to prepare states that have sufficiently precise simultaneous values of incompatible observables.
The most well-known form of this statement is the position-momentum uncertainty relation $\Delta x\Delta p \ge \hbar/2$, which places a lower bound on the product of standard deviations of the position and momentum observables, for a particle in any possible quantum state.
This relation was first formalised by Kennard~\cite{Kennard:1927oa} during the formative years of quantum mechanics following Heisenberg's discussion of his ``uncertainty principle''~\cite{Heisenberg:1927zh}. 

This ``uncertainty relation'' was quickly generalised by Robertson~\cite{Robertson:1929ee} to arbitrary pairs of incompatible (i.e.,\ non-commuting) observables $A$ and $B$ into what is now the textbook uncertainty relation.
Let $A$ and $B$ be two observables and $[A,B]=AB-BA$ their commutator.
If the standard deviations $\Delta A$ and $\Delta B$ for a system in the state $\rho$ are defined as
\begin{equation}
	\Delta A = \sqrt{\braket{A}^2-\braket{A^2}},\qquad\qquad \Delta B = \sqrt{\braket{B}^2-\braket{B^2}},
\end{equation}
where $\braket{\,\cdot\,}=\Tr[\rho\,\cdot \,]$, then Robertson's uncertainty relation can be expressed as 
\begin{equation}\label{eqn:RobertsonsUR}
\Delta A \, \Delta B \ge \left|\Braket{{\textstyle \frac{1}{2i}} [A,B]}\right|.
\end{equation}

\noindent
These uncertainty relations express a quantitative statement about the measurement statistics for $A$ and $B$ when they are \emph{measured many times, separately, on identically prepared quantum systems}.
Such relations are hence sometimes called \emph{preparation uncertainty relations}, since they propose fundamental limits on the measurement statistics for any state preparation.

This is in contrast to Heisenberg's original discussion of his uncertainty principle which he expressed as the inability to \emph{simultaneously measure} incompatible observables with arbitrary accuracy.
As such, quantum uncertainty relations have a long history of being misinterpreted as statements about joint measurements.
It is only much more recently that progress has been made in formalising \emph{measurement uncertainty relations} that quantify measurement disturbance in this way, although there continues to be some debate as to the appropriate measure of measurement (in)accuracy and of disturbance~\cite{Ozawa:2003fh,Hall:2004eq,Branciard:2013hb,Busch:2013mb,Ozawa:2013ss,Dressel:2014fx,Busch:2014ts}.

The recent interest in measurement uncertainty relations has highlighted an oft-overlooked aspect of Robertson's inequality~\eqref{eqn:RobertsonsUR}: its \emph{state dependence}.
Indeed, the right-hand side of Eq.~\eqref{eqn:RobertsonsUR} depends on the expectation value $\left|\Braket{\frac{1}{2i} [A,B]}\right|$, which itself depends on the state $\rho$ of the system and may be zero even for non-commuting $A$ and $B$.
To illustrate this, consider a spin-$\frac{1}{2}$ particle and measurement of Pauli-spin operators $\sigma_x,\sigma_y,\sigma_z$.
Robertson's inequality gives us 
\begin{equation}
\Delta \sigma_x \, \Delta \sigma_y \ge \left|\Braket{{\textstyle \frac{1}{2i}} [\sigma_x,\sigma_y]}\right|=|\Braket{\sigma_z}|,
\end{equation}
where the right-hand side is zero for any state $\rho = \frac12(\vect{1}+r_x \sigma_x + r_y \sigma_y)$ (with $r_x^2 + r_y^2 \le 1$, and with $\vect{1}$ denoting the identity operator), even though for $\rho \neq \frac12 (\vect{1}\pm\sigma_x)$ and $\rho \neq \frac12 (\vect{1}\pm\sigma_y)$ both $\Delta \sigma_x$ and $\Delta \sigma_y$ are strictly positive.

Robertson's inequalities, like other historical inequalities such as those due to Schrödinger~\cite{Schrodinger:1930gf}, therefore often tell us little about the accuracy with which one can prepare a state with respect to two incompatible observables $A$ and $B$.
There are two distinct issues with such inequalities: their triviality for certain states, and the state dependence itself.
The first issue can, in some cases, be avoided by considering more complicated expressions or different measures of incompatibility~\cite{Kaniewski:2014aa,Maccone:2014kr}.
This approach can be used to give non-trivial state-dependent uncertainty relations, which have the property that they can be experimentally verified without knowing the observables $A$ and $B$, and thus are of interest for device-independent cryptography~\cite{Kaniewski:2014aa}.
However, one may equally be interested in knowing how accurately one can prepare a state with respect to two given incompatible observables $A$ and $B$ -- that is, in characterising the ``minimum uncertainty'' states of a system.
In such a situation, one ideally wants an uncertainty relation that depends on the state of the system only via the (operationally defined) measures of uncertainty, which  ensures that the relation is an operational statement constraining the uncertainties directly, and can be evaluated without prior knowledge of the system's state.
This is indeed the case, for example, with the position-momentum uncertainty relation described earlier.
It thus makes sense to look for tighter, \emph{state-independent} relations capable of addressing these issues, and it is this situation we tackle in this paper.

\subsection{Entropic uncertainty relations}

It has long been known that an alternative form of uncertainty relation can be given by considering the entropies of the observables rather than their standard deviations~\cite{Hirschman:1957ul}.
It was Deutsch who first realised that such entropic uncertainty relations can be used to provide state-independent relations~\cite{Deutsch:1983ai}, and thus avoid the problems with traditional relations discussed above.

Rather than placing lower bounds on the product of variances of two observables $A$ and $B$, entropic uncertainty relations generally place bounds on the sum of the entropies of $A$ and $B$.
Although many different entropies can be used to formulate such inequalities, perhaps the most well known one, due to Maassen and Uffink~\cite{Maassen:1988vl}, makes use of the Shannon entropy.
If $A$ has a spectral decomposition $A=\sum_{i=1}^{d} a_i P_i$ in terms of projectors $P_i$ (with $\sum_{i=1}^{d} P_i = \vect{1}$) then the Shannon entropy of $A$ for a system in a state $\rho$ is defined as 
\begin{equation}\label{eqn:def_HA}
H(A)=-\sum_{i=1}^d \Tr[\rho\, P_i]\log\big(\Tr[\rho\, P_i]\big),
\end{equation}
which can be seen as a measure of the uncertainty in $A$ for the state $\rho$.
Following the information-theoretical convention we take the logarithms in base 2, although any base can be used as long as the choice is consistent.

Maassen and Uffink's inequality can then be stated as 
\begin{equation}\label{eqn:EntrRelnMU}
H(A) + H(B) \ge -2\log c,
\end{equation}
where $c=\max_{i,j}|\braket{a_i\mid b_j}|$ is the maximum overlap between the eigenvectors $\ket{a_i}$ and $\ket{b_j}$ of $A$ and $B$.
The lower bound is thus state \emph{independent}, and depends only on the unitary operator connecting the eigenbases of $A$ and $B$.
Although this is a significant conceptual improvement over state-dependent relations, it is not optimal since the bound cannot be saturated except for well-chosen $A$ and $B$.

Many variations and improvements on this relation have been found (see Ref.~\cite{Coles:2015ef} for a recent review) but, without imposing further restrictions, tight bounds have proved elusive.
In the two dimensional case (i.e., for qubit measurements), several papers have improved on this bound~\cite{Garrett:1990bf,Sanchez-Ruiz:1998by,Ghirardi:2003ez} to determine the optimal lower bound on the sum $H(A)+H(B)$ for arbitrary qubit observables $A$ and $B$.
This result was further generalised to a range of higher dimensional systems in Ref.~\cite{Vicente:2008oq}.
However, these bounds are not tight in the sense that, although there exist states that saturate the bound, there exist pairs of entropy values $(H(A),H(B))$ that satisfy the uncertainty relation but are not permitted by quantum mechanics.

In order to fully characterise the obtainable uncertainties, it is thus necessary to consider functions of $H(A)$ and $H(B)$ beyond their sum;
indeed, there is no \emph{a priori} reason why one should only consider entropic uncertainty relations based on the sum $H(A)+H(B)$.
Much more recent work~\cite{Abdelkhalek:2015cr} has made progress in this direction working with more general Rényi entropies, and presents several conjectures and numerical results beyond two dimensions.

\subsection{State-independent uncertainty relations for standard deviations}

The growth of interest in measurement uncertainty relations has prompted renewed interest in the possibility of state-independent uncertainty relations for the standard deviations of observables, rather than entropic relations~\cite{Huang:2012aa,Li:2015wf}.
Particular attention has been devoted to understanding the simplest case of qubit uncertainty relations, and we continue this line of research in this paper.

It will be convenient to use the Bloch sphere representation.
For the two-dimensional case of qubits, an arbitrary $\pm 1$-valued observable -- a ``Pauli observable'' -- $A$ can be written $A=\vect{a}\cdot \bm{\sigma}$ where $\bm{\sigma}=(\sigma_x,\sigma_y,\sigma_z)^\T$ and $\vect{a}$ is a unit vector.
Similarly, an arbitrary state $\rho$ can be written $\rho=\frac12 (\vect{1}+\vect{r}\cdot \bm{\sigma})$, where $|\vect{r}|=1$ if $\rho$ is a pure state and $|\vect{r}|<1$ for mixed states.

Busch et al.~\cite{Busch:2014fu} proved two state-independent uncertainty relations for arbitrary Pauli measurements $A=\vect{a}\cdot \bm{\sigma}$ and $B=\vect{b}\cdot \bm{\sigma}$, showing that
\begin{equation}\label{eqn:Busch1}\Delta A + \Delta B \ge |\vect{a}\times \vect{b}|\end{equation}
and
\begin{equation}\label{eqn:Busch2}(\Delta A)^2 + (\Delta B)^2 \ge 1 - |\vect{a}\cdot \vect{b}|.\end{equation}
Although these state-independent relations can be saturated by certain states, neither are tight in the sense discussed in the previous section.
That is, there exist pairs of values $(\Delta A, \Delta B)$ which are allowed by the relations but not realisable by any quantum state $\rho$.

In this paper we aim precisely to provide \emph{tight} state-independent uncertainty relations, so as to fully characterise the set of allowed values of $(\Delta A,\Delta B)$, for all possible states $\rho$ -- which we shall call the ``\emph{uncertainty region}''.
In Ref.~\cite{Dammeier:2015zt} this is done for the case that $|\vect{a}\cdot\vect{b}|=0$ (i.e., for orthogonal spin directions), and it was shown that the bound~\eqref{eqn:Busch2} is tight in this specific case. 
The authors further gave a tight relation for three orthogonal spin directions, as well as generalisations, both for pairs and triples of orthogonal spin observables, to higher dimensions, including the asymptotic behaviour, although these higher-dimensional results are not tight.

The restriction to orthogonal spin observables, however, is a rather strong one.
The general case for arbitrary pairs of qubit measurements was completely characterised in~\cite{Li:2015wf} using geometric methods, in particular the fact that $|\theta_{ra}-\theta_{rb}|\le\theta_{ab}\le\theta_{ra}+\theta_{rb}$, where $\theta_{ra}$ is the angle between $\vect{r}$ and $\vect{a}$, etc.
This method leads to the state-independent uncertainty relation\footnote{Eq.~\eqref{eqn:Li-QiaoReln} was proven in~\cite{Li:2015wf} for pure states (and written explicitly for $\vect{a}\cdot\vect{b} \ge 0$ only), and another version was given, for any fixed value of $|\vect{r}|$. With our equivalent version given in Eq.~\eqref{eqn:SIReln2obs} below, it is straightforward to see that~\eqref{eqn:Li-QiaoReln} also holds for mixed states.}
\begin{equation}
\label{eqn:Li-QiaoReln}\Delta A \, \Delta B \ \ge \ \left| \sqrt{1-(\Delta A)^2}\sqrt{1-(\Delta B)^2}- |\vect{a}\cdot\vect{b}| \right|.
\end{equation}
Although the authors discuss some explicit instances of three-observable relations, their method does not readily lead to a generalised form of three-observable relations.

In this paper we present a different approach to deriving tight state-independent qubit uncertainty relations.
This approach not only leads to a simpler derivation of the relation~\eqref{eqn:Li-QiaoReln}, but is immediately generalisable to give relations for three or more arbitrary observables.
It can also be used for other uncertainty measures beyond standard deviations, such as entropic measures.
This offers a unified approach to completely characterising the possible state-independent uncertainty relations in two-dimensional Hilbert space.

\section{A unified approach to qubit uncertainty relations}
\label{sec:uncertainties}

We present here a general method to derive tight state-independent uncertainty relations for qubits. 
The approach is based on the fact that, in the case of a binary measurement $A$, the expectation value $\braket{A}$ contains all the information about the uncertainty in the measurement: for instance, both $\Delta A$ and $H(A)$ can be expressed as simple functions of $\braket{A}$. 
We will thus start by giving relations characterising the set of allowed values $(\braket{A}, \braket{B})$, before translating these relations into ones in terms of standard deviations or entropies.

A formal characterisation of the set of possible uncertainty values was given by Kaniewski \emph{et al.}\ in Ref.~\cite{Kaniewski:2014aa} for the more general case of binary-valued measurements in arbitrary dimensional Hilbert spaces.
Although the characterisation they give, which is formulated in terms of expectation values of anticommutators, leads to the results we present in this section (specifically, Lemmas~\ref{lemma:generalForm}--\ref{lemma_nobs} below, albeit in a different mathematical framework), Kaniewski \emph{et al.}\ use it to derive state-\emph{dependent} entropic uncertainty relations for this generalised scenario which bound the sum of the entropies considered, and are hence not tight in the sense we consider.
In this paper our goal is instead to use this characterisation to derive general forms of tight, state-independent uncertainty relations for qubits.

For simplicity we will only consider below Pauli measurements (with eigenvalues $\pm 1$), although, as we will discuss in Section~\ref{sec:discussion}, our results can straightforwardly be generalised to any projective measurements and even, with a little effort, to binary-valued positive-operator valued measures (POVMs). 
We shall start with the case of two observables, before generalising to any number of measurements.

\subsection{For two observables}

Let us first consider two arbitrary Pauli observables $A=\vect{a}\cdot \bm{\sigma}$ and $B=\vect{b}\cdot \bm{\sigma}$, and define the matrix
\begin{equation}\label{eqn:DefM_2}
M = \begin{pmatrix}\vect{a}^\T\\\vect{b}^\T\end{pmatrix}
\end{equation}
with $\vect{a}$ and $\vect{b}$ as rows representing the measurement directions.
In the Bloch sphere representation we have $\braket{A}=\Tr[\rho\, A]=\vect{a}\cdot \vect{r}$ and $\braket{B}=\Tr[\rho\, B]=\vect{b}\cdot \vect{r}$, so that
\begin{equation}\label{eqn:uDefn}
M\vect{r}=\begin{pmatrix}\vect{a}\cdot\vect{r}\\\vect{b}\cdot\vect{r}\end{pmatrix}=\begin{pmatrix}\braket{A}\\\braket{B}\end{pmatrix} \coloneqq \vect{u} \,.
\end{equation}

Although $M$ is not invertible, one can always find the Moore-Penrose pseudoinverse $M^+$ such that $M^+M$ is an orthogonal projection onto the range of $M^\T$, that is, the subspace spanned by $\{\vect{a},\vect{b}\}$~\cite{Horn:1985aa}.
This implies the following crucial Lemma, which will be the basis of the uncertainty relations we derive in the following sections.\footnote{This Lemma is equivalent, for qubit measurements, to the ``ellipsoid condition'' given in Ref.~\cite{Kaniewski:2014aa}.}
\begin{Lemma}\label{lemma:generalForm}
	For any pair of Pauli observables $A=\vect{a}\cdot \bm{\sigma}$ and $B=\vect{b}\cdot \bm{\sigma}$ with $M$ and $\vect{u}$ as defined in Eqs.~\eqref{eqn:DefM_2} and~\eqref{eqn:uDefn}, and $M^+$ the pseudoinverse of $M$, every quantum state $\rho = \frac12 (\vect{1}+\vect{r}\cdot \bm{\sigma})$ satisfies
	\begin{equation}\label{eqn:PseudoInvSD}
		|M^+\vect{u}| = |(M^+M)\vect{r}| \le |\vect{r}|,
	\end{equation}
	where equality is obtained if and only if $\vect{r}$ lies in $\Span \{\vect{a},\vect{b}\}$.
\end{Lemma}

In a more explicit form, this inequality implies the following relation for the two expectation values $\braket{A}$ and $\braket{B}$:

\begin{Lemma}\label{lemma_2obs}
For any pair of Pauli observables $A=\vect{a}\cdot \bm{\sigma}$ and $B=\vect{b}\cdot \bm{\sigma}$, every quantum state $\rho = \frac12 (\vect{1}+\vect{r}\cdot \bm{\sigma})$ satisfies
\begin{equation}\label{eqn:SIReln2obs_avg}
		 \left|\ \braket{A}\vect{a} - \braket{B}\vect{b} \ \right|^2
		 \ \braket{A}^2 + \braket{B}^2 - 2 \, (\vect{a}\cdot\vect{b}) \braket{A} \braket{B}
		\  \le \ \big(1-(\vect{a}\cdot\vect{b})^2 \big) \, |\vect{r}|^2
		\ \le \ 1 - (\vect{a}\cdot\vect{b})^2 = |\vect{a}\times\vect{b}|^2 \,.
\end{equation}
In the case where $|\vect{a}\cdot\vect{b}| < 1$, the first and second inequalities are saturated if and only if $\vect{r}\in\Span\{\vect{a},\vect{b}\}$ and if and only if $\rho$ is a pure state, respectively.
\end{Lemma}

\begin{proof}
Let us choose a basis in the Bloch sphere such that $\vect{a}$ and $\vect{b}$ can be written $\vect{a} = \big(\sqrt{\frac{1+\vect{a}\cdot\vect{b}}{2}}, \sqrt{\frac{1-\vect{a}\cdot\vect{b}}{2}}, 0\big)^\T$ and $\vect{b} = \big(\sqrt{\frac{1+\vect{a}\cdot\vect{b}}{2}}, -\sqrt{\frac{1-\vect{a}\cdot\vect{b}}{2}}, 0\big)^\T$, and assume first that $|\vect{a}\cdot\vect{b}| < 1$.
Then $M$ has linearly independent rows and its pseudoinverse can be obtained using the relation $M^+=M^\T(MM^\T)^{-1}$. 
We thus have
\begin{equation}
M^+ = \frac{1}{\sqrt{1-(\vect{a}\cdot\vect{b})^2}}
\begin{pmatrix}
	\sqrt{\frac{1-\vect{a}\cdot\vect{b}}{2}} & \sqrt{\frac{1-\vect{a}\cdot\vect{b}}{2}} \\
	\sqrt{\frac{1+\vect{a}\cdot\vect{b}}{2}} & - \sqrt{\frac{1+\vect{a}\cdot\vect{b}}{2}} \\
	0 & 0\\
\end{pmatrix} \quad \text{and} \quad
M^+M = 
\begin{pmatrix}
	1 & 0 & 0\\
	0 & 1 & 0\\
	0 & 0 & 0\\
\end{pmatrix},
\end{equation}
where we see that $M^+M$ is a projection onto $\Span\{\vect{a},\vect{b}\}$, as expected.
The first inequality in Eq.~\eqref{eqn:SIReln2obs_avg} is then obtained by squaring Eq.~\eqref{eqn:PseudoInvSD}, multiplying it by $\big(1-(\vect{a}\cdot\vect{b})^2 \big)$, and expanding; like~\eqref{eqn:PseudoInvSD}, it is saturated when $\vect{r}\in\Span\{\vect{a},\vect{b}\}$. The second inequality follows directly from $|\vect{r}| \le 1$, and is saturated when $\rho$ is a pure state.

In the (trivial) case where $\vect{a}\cdot\vect{b} = \pm 1$ (i.e., $A = \pm B$), one can easily check that the left-hand side of Eq.~\eqref{eqn:SIReln2obs_avg} is $0$ so that the relation still holds.
\end{proof}

Relation~\eqref{eqn:SIReln2obs_avg}, as also noted in Ref.~\cite{Kaniewski:2014aa}, shows that the set of allowed values for $\braket{A}$ and $\braket{B}$ forms an ellipse in the $(\braket{A}, \braket{B})$-plane, as depicted in Fig.~\ref{fig:AvgRelns} for $\vect{a}\cdot\vect{b} = 0$ and $\vect{a}\cdot\vect{b}=\frac12$ (cf. Fig.~1 in Ref.~\cite{Kaniewski:2014aa}). 
As can be seen, this ellipse becomes a circle for $\vect{a}\cdot\vect{b} = 0$, and degenerates into the line segment given by $\braket{A} = \pm \braket{B}$ when $\vect{a}\cdot\vect{b} = \pm 1$.
Note also that the first inequality in~\eqref{eqn:SIReln2obs_avg} characterises concentric ellipses for fixed maximal values of $|\vect{r}|$ (see Figure~\ref{fig:AvgRelns}).

\begin{figure}[t]
	\begin{center}
		\begin{tabular}{ccc}
			\includegraphics[width=0.47\textwidth]{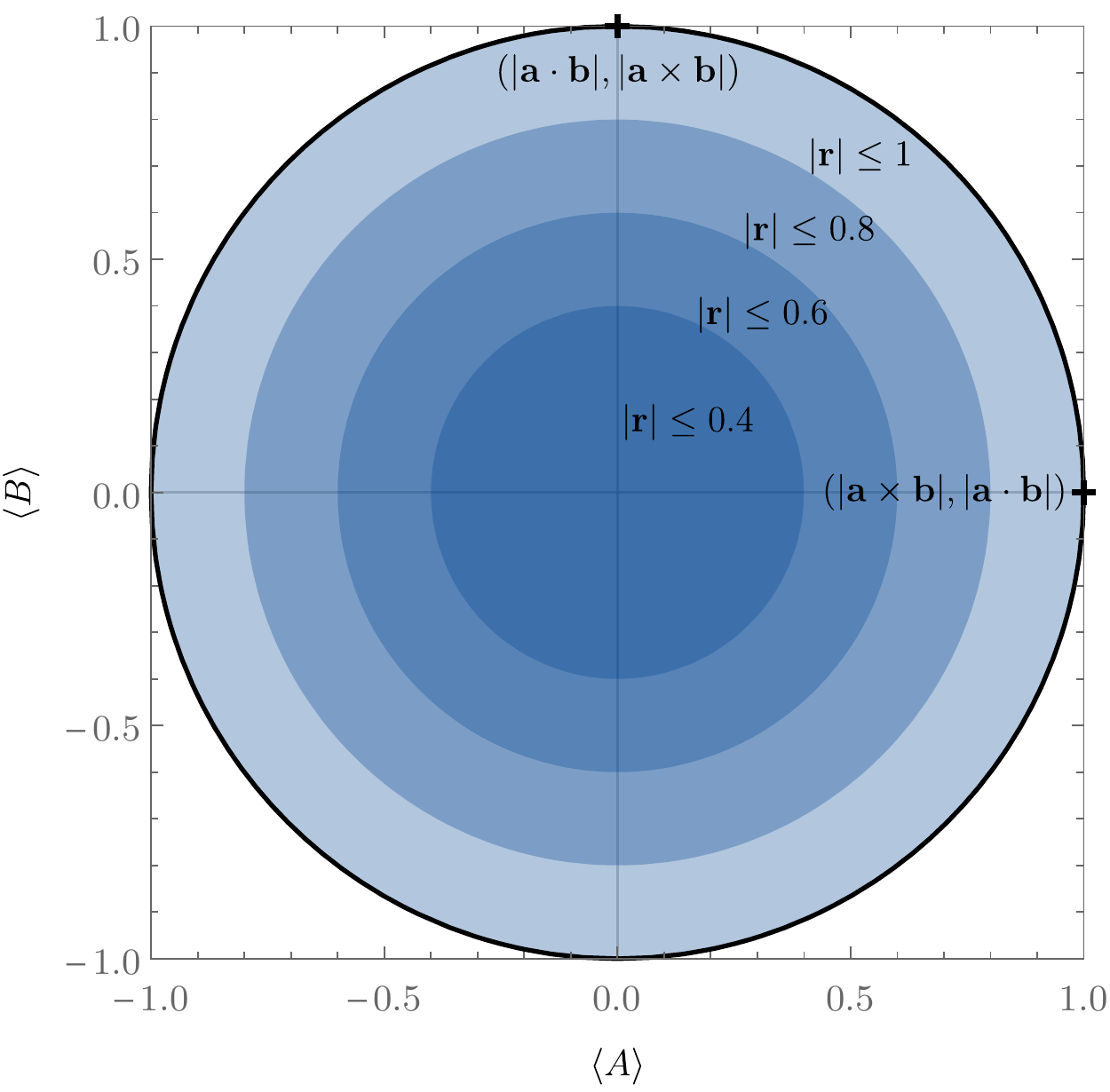} &
			 &
			\includegraphics[width=0.47\textwidth]{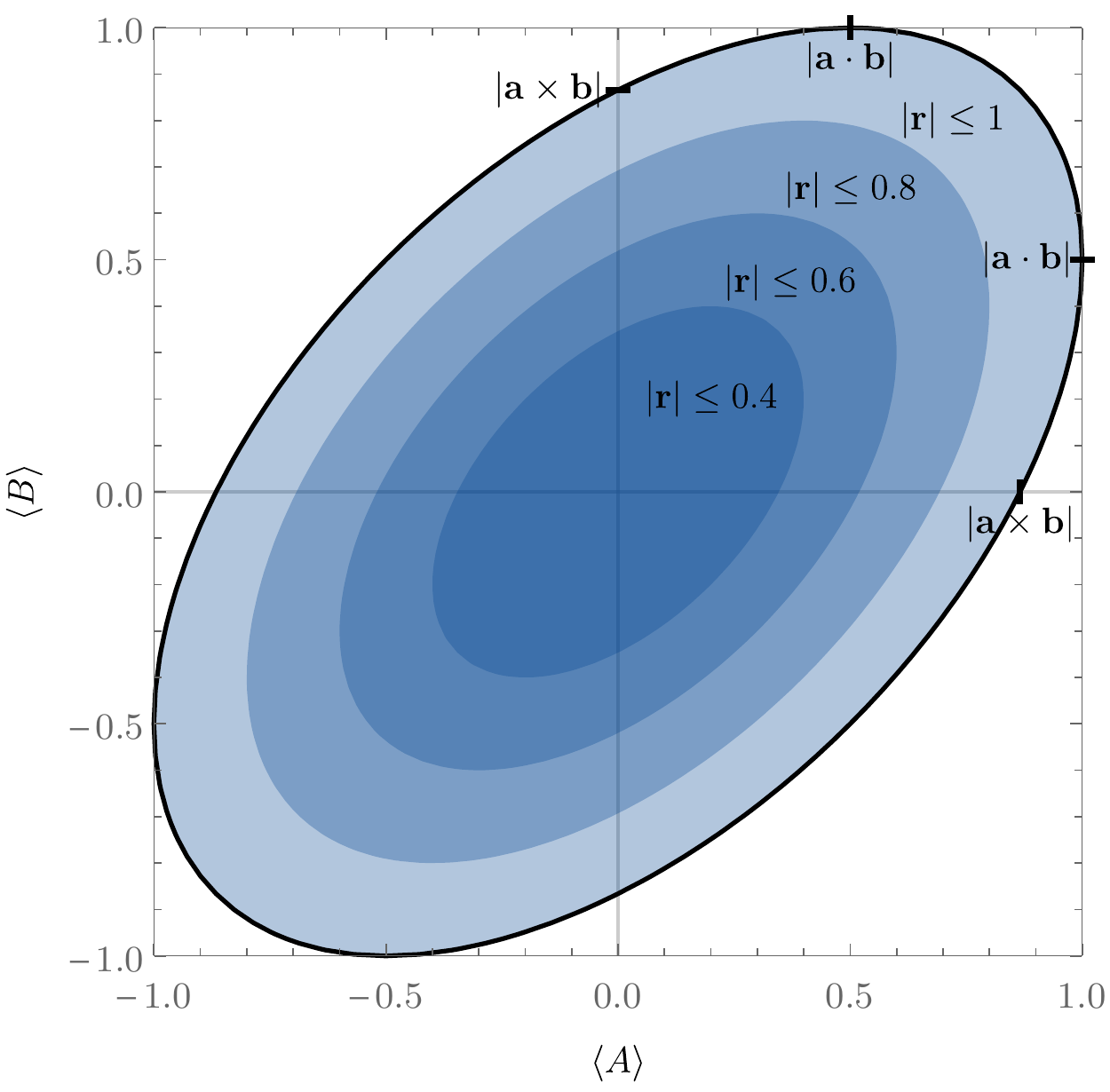}\\
			(a) & & (b)\\
		\end{tabular}
	\end{center}
	\caption{The ellipse delimited by the thick black line corresponds to the set of all permitted values $(\braket{A},\braket{B})$ for two fixed Pauli operators $A=\vect{a}\cdot \bm{\sigma}$ and $B=\vect{b}\cdot \bm{\sigma}$ with (a) $\vect{a}\cdot\vect{b} = 0$ and (b) $\vect{a}\cdot\vect{b} = \frac12$ for all possible states $\rho$, defined by the relation~\eqref{eqn:SIReln2obs_avg}. The values $(\braket{A},\braket{B})$ on its boundary saturate Eq.~\eqref{eqn:SIReln2obs_avg}. The darker concentric ellipses represent the regions of permitted values for mixed states with bounded Bloch vector norms $|\vect{r}| \le 0.8, 0.6$, and $0.4$, characterised using the first inequality in Eq.~\eqref{eqn:SIReln2obs_avg}.}
	\label{fig:AvgRelns}
\end{figure}

We emphasise that the relation~\eqref{eqn:SIReln2obs_avg} is tight. To verify (for the nontrivial case that $|\vect{a}\cdot\vect{b}| < 1$) that any pair $(\braket{A},\braket{B})$ satisfying Eq.~\eqref{eqn:SIReln2obs_avg} can be attained, take for instance
\begin{equation}
		\vect{r} = \frac{1}{1-(\vect{a}\cdot\vect{b})^2} \Big[ \braket{A} \big( \vect{a} - (\vect{a}\cdot\vect{b}) \, \vect{b} \big) + \braket{B} \big( \vect{b} - (\vect{a}\cdot\vect{b}) \, \vect{a} \big) \Big] \, .
\end{equation}
It gives the desired values for $\braket{A}$ and $\braket{B}$ and is indeed a valid Bloch vector -- that is, its norm is at most $1$ -- if and only if Eq.~\eqref{eqn:SIReln2obs_avg} is satisfied. 
It clearly lies in the $\vect{a}\vect{b}$-plane and thus saturates the first inequality in~\eqref{eqn:SIReln2obs_avg}.
To saturate both inequalities in~\eqref{eqn:SIReln2obs_avg} for a given value of $\braket{A}$, one can take
\begin{equation}
		\vect{r}_\pm = \braket{A} \vect{a} \pm \sqrt{\frac{1-\braket{A}^2}{1-(\vect{a}\cdot\vect{b})^2}} \big( \vect{b} - (\vect{a}\cdot\vect{b}) \, \vect{a} \big) \, ,
\end{equation}
which characterises the pure states ($|\vect{r}_\pm | = 1$) in the $\vect{a}\vect{b}$-plane and gives the desired value for $\braket{A}$.

\subsection{Generalisation to more observables}

Our approach, based on Lemma~\ref{lemma:generalForm}, generalises easily to more than two observables. 
Remarkably, when the various observables $A_i$ span the whole Bloch sphere, it provides \emph{exact relations} on the expectation values $\braket{A_i}$, rather than just inequalities.

More specifically, let us consider $n$ observables $A_1=\vect{a}_1\cdot \bm{\sigma},\dots,A_n=\vect{a}_n\cdot \bm{\sigma}$ and define $M=(\vect{a}_1^T,\dots,\vect{a}_n^\T)^\T$, $M^+$ its pseudoinverse, and $\vect{u} = (\braket{A_1},\dots,\braket{A_n})^\T$.
Then the relation~\eqref{eqn:PseudoInvSD} of Lemma~\ref{lemma:generalForm} holds unchanged, with equality if and only if $\vect{r}$ lies in $\Span \{\vect{a}_1,\dots,\vect{a}_n\}$. This straightforwardly implies (after squaring Eq.~\eqref{eqn:PseudoInvSD}) the following relation:

\begin{Lemma}\label{lemma_nobs}
For any $n$ Pauli observables $A_1,\dots,A_n$ with Bloch vectors $\vect{a}_1,\dots,\vect{a}_n$ spanning the whole Bloch sphere, every quantum state $\rho = \frac12 (\vect{1}+\vect{r}\cdot \bm{\sigma})$ satisfies the relation
\begin{equation}\label{eqn:SIReln_nobs_avg}
	\sum_{1 \le i,j \le n} m_{i,j} \braket{A_i} \braket{A_j} = |\vect{r}|^2 \le 1,
\end{equation}
where the coefficients $m_{i,j}$ are the elements of the symmetric matrix $(M^+)^\T M^+$, with $M^+$ the pseudoinverse of the matrix $M=(\vect{a}_1^T,\dots,\vect{a}_n^\T)^\T$.

If the $n$ Bloch vectors $\vect{a}_1,\dots,\vect{a}_n$ do not span the whole Bloch sphere, then the equality in Eq.~\eqref{eqn:SIReln_nobs_avg} must be replaced by an inequality (with $|\vect{r}|^2$ upper-bounding the left-hand side), which is saturated if and only if $\vect{r} \in \Span \{\vect{a}_1,\dots,\vect{a}_n\}$.
\end{Lemma}

This relation shows that, in the general case, the set of allowed values for $(\braket{A_1},\dots,\braket{A_n})$ lie on an $n$-dimensional ellipsoid~\cite{Kaniewski:2014aa}; when at least three of the Bloch vectors are linearly independent (i.e., when the $n$ Bloch vectors span the whole Bloch sphere), all pure states give points on the surface of the ellipsoid, while mixed states give interior points.

Although this relation is satisfied for any state $\rho$, for $n>3$ observables (or, more generally, if $n$ exceeds the dimension of $\Span \{\vect{a}_1,\dots,\vect{a}_n\}$) it is not tight since there exist values $(\braket{A_1},\dots,\braket{A_n})\coloneqq\vect{u}$ which satisfy it, but are not obtainable by any quantum state $\rho$.
Specifically, $\vect{u}$ is realisable if and only if there exists a quantum state with Bloch vector $\vect{r}$ such that $M\vect{r}=\vect{u}$.
For $n=3$ observables with $\vect{a}_1,\vect{a}_2,\vect{a}_3$ linearly independent, one has $MM^+=\vect{1}$ and one may simply take $\vect{r}=M^+\vect{u}$, so~\eqref{eqn:SIReln_nobs_avg} is tight in this case.
However, if the Bloch vectors $\vect{a}_1,\dots\vect{a}_n$ are not linearly independent (as, in particular, is the case for $n > 3$) then such an $\vect{r}$ exists if and only if $MM^+\vect{u}=\vect{u}$.

One can understand this by noting that, once $\braket{A_1},\braket{A_2},\braket{A_3}$ are determined (assuming, without loss of generality, that $\vect{a}_1,\vect{a}_2,\vect{a}_3$ are linearly independent) then this uniquely determines each $\braket{A_i}$ for $i>3$.
Thus, only three of the expectation values can be considered free variables, whereas Eq.~\eqref{eqn:SIReln_nobs_avg} has $n-1$ free variables.
The requirement that $MM^+\vect{u}=\vect{u}$ expresses this further constraint. 
Although Eq.~\eqref{eqn:SIReln_nobs_avg} is thus not, by itself, a tight relation for $n\ge 4$, 
it becomes so when supplemented by the further constraint that the expectation values $(\braket{A_1},\dots,\braket{A_n})\coloneqq\vect{u}$ also satisfy $MM^+\vect{u}=\vect{u}$.

Let us provide some explicit examples of relations based on Lemma~\ref{lemma_nobs}. 
Consider first the case of $n = 3$ Pauli observables $A=\vect{a}\cdot \bm{\sigma}$, $B=\vect{b}\cdot \bm{\sigma}$ and $C=\vect{c}\cdot \bm{\sigma}$ with linearly independent Bloch vectors: we obtain, after multiplication by 
$V^2=(\, \vect{a}\cdot(\vect{b}\times\vect{c})\, )^2>0$ 
the relation\footnote{To calculate the matrix $(M^+)^\T M^+$ and obtain Eq.~\eqref{eqn:SIReln_3obs_avg}, one can for instance parametrize $\vect{a}$ and $\vect{b}$ as in the proof of Lemma~\ref{lemma_2obs}, and define $\vect{c} = \frac{1}{\sqrt{1-(\vect{a} \cdot \vect{b})^2}} \big( ( \vect{a} \cdot \vect{c} + \vect{b} \cdot \vect{c} ) \sqrt{\frac{1-\vect{a}\cdot\vect{b}}{2}}, ( \vect{a} \cdot \vect{c} - \vect{b} \cdot \vect{c} ) \sqrt{\frac{1+\vect{a}\cdot\vect{b}}{2}}, V\big)^\T$.}

\begin{eqnarray}\label{eqn:SIReln_3obs_avg}
	&& \left|\ (\vect{b}\times\vect{c})\braket{A} + (\vect{c}\times\vect{a})\braket{B} + (\vect{a}\times\vect{b})\braket{C} \ \right|^2 \notag \\
	&& \qquad  = \ \left| \vect{b}\times\vect{c} \right|^2 \braket{A}^2 + \left| \vect{a}\times\vect{c} \right|^2 \braket{B}^2 + \left| \vect{a}\times\vect{b} \right|^2 \braket{C}^2 + 2 \, (\vect{b}\times\vect{c})\cdot(\vect{c}\times\vect{a}) \braket{A} \braket{B} \nonumber \\
	&& \qquad \qquad \ + \,  2 \, (\vect{b}\times\vect{c})\cdot(\vect{a}\times\vect{b}) \braket{A} \braket{C} + 2 \, (\vect{c}\times\vect{a})\cdot(\vect{a}\times\vect{b}) \braket{B} \braket{C} \ = \ V^2 \, |\vect{r}|^2 \ \le \ V^2 \,. \qquad \quad
\end{eqnarray}

\noindent
This relation \emph{is} tight, and the quantum state with Bloch vector $\vect{r}=M^+(\braket{A},\braket{B},\braket{C})^\T$ has the required expectation values for any $(\braket{A},\braket{B},\braket{C})$ satisfying~\eqref{eqn:SIReln_3obs_avg}.
In the special case of three orthogonal measurements ($\vect{a} \cdot \vect{b} = \vect{a} \cdot \vect{c} = \vect{b} \cdot \vect{c} = 0$, $V^2=1$), we thus find the well-known relation

\begin{equation}
	\braket{A}^2 + \braket{B}^2 + \braket{C}^2 \ = \ |\vect{r}|^2 \ \le \ 1 \,.
\end{equation}

Consider as another example $n = 4$ observables with Bloch vectors pointing to the vertices of a regular tetrahedron. 
In that case we find\footnote{One can use here for instance the parametrisation $\vect{a}_1 = (1,1,1)^\T/\sqrt{3}$, $\vect{a}_2 = (1,-1,-1)^\T/\sqrt{3}$, $\vect{a}_3 = (-1,1,-1)^\T/\sqrt{3}$, $\vect{a}_4 = (-1,-1,1)^\T/\sqrt{3}$.} 
\begin{equation}\label{eqn:tetrahedronExpectations}
		3 \sum_{1 \le i \le 4} \braket{A_i}^2 - \sum_{1 \le i \neq j \le 4} \braket{A_i}\braket{A_j} \ = \ \frac{16}{3} \, |\vect{r}|^2 \ \le \ \frac{16}{3} \, .
\end{equation}
For any quantum state $\rho$, the expectation values for these observables must further satisfy $\braket{A_4}=-\braket{A_1}-\braket{A_2}-\braket{A_3}$.
Thus, Eq.~\eqref{eqn:tetrahedronExpectations} can be made tight by further imposing this constraint or simply replacing $\braket{A_4}$ by this expression in the relation.

\section{Uncertainty relations in terms of standard deviations}

The relations for the expectation values presented in the previous section now allow us to proceed with the main aim of the paper and derive tight state-independent uncertainty relations in terms of standard deviations. 
To do so, we note that any Pauli operator $A$ satisfies $A^2 = \vect{1}$, so that $\braket{A^2} = 1$ and therefore the standard deviation $\Delta A$ is simply related to the expectation value $\braket{A}$ by
\begin{equation}\label{eqn:StdDev}
	(\Delta A)^2 = 1 - \braket{A}^2 \quad \text{and} \quad \braket{A} = \pm \sqrt{1 - (\Delta A)^2} \,.
\end{equation}
Note that because of the $\pm$ sign above, some care needs to be taken when applying the previous relations to derive uncertainty relations for standard-deviations.

\subsection{Uncertainty relation for two Pauli observables}

Let us start again with two Pauli observables. Using 
$(\vect{a}\cdot\vect{b}) \braket{A} \braket{B} \le |(\vect{a}\cdot\vect{b}) \braket{A} \braket{B}| = |\vect{a}\cdot\vect{b}|\sqrt{1-(\Delta A)^2}\sqrt{1-(\Delta B)^2}$ and reordering the terms, the following relation follows directly from Lemma~\ref{lemma_2obs}:

\begin{Theorem}\label{thm:SIReln2obs}
	For any pair of Pauli observables $A=\vect{a}\cdot \bm{\sigma}$ and $B=\vect{b}\cdot \bm{\sigma}$, every quantum state $\rho = \frac12 (\vect{1}+\vect{r}\cdot \bm{\sigma})$ satisfies the state-independent uncertainty relation
	\begin{equation}\label{eqn:SIReln2obs}
		(\Delta A)^2 + (\Delta B)^2 + 2 \, |\vect{a}\cdot\vect{b}|\sqrt{1-(\Delta A)^2}\sqrt{1-(\Delta B)^2}
		\ \ge \ 2 - \big(1-(\vect{a}\cdot\vect{b})^2 \big) \, |\vect{r}|^2
		\ \ge \ 1 + (\vect{a}\cdot\vect{b})^2.
	\end{equation}
	In the case where $|\vect{a}\cdot\vect{b}| < 1$, the first inequality above is saturated if and only if $\vect{r}\in\Span\{\vect{a},\vect{b}\}$ and $(\vect{a}\cdot\vect{b}) \braket{A} \braket{B} \ge 0$, and the second one if and only if $\rho$ is a pure state.
\end{Theorem}

One can easily check that the relation~\eqref{eqn:SIReln2obs} is equivalent to Eq.~\eqref{eqn:Li-QiaoReln} and to the relation given in Ref.~\cite{Li:2015wf} for a fixed value of $|\vect{r}|$.
The uncertainty region it defines\footnote{It may be interesting to note that, in the $((\Delta A)^2,(\Delta B)^2)$-plane, this region corresponds to the convex hull of an ellipse and the point $(1,1)$.}
is shown in Figure~\ref{fig:UncertanityRelns} for the cases where $\vect{a}\cdot\vect{b} = 0$ and $|\vect{a}\cdot\vect{b}| = \frac{1}{2}$.
\begin{figure}[t]
	\begin{center}
		\begin{tabular}{ccc}
			\includegraphics[width=0.47\textwidth]{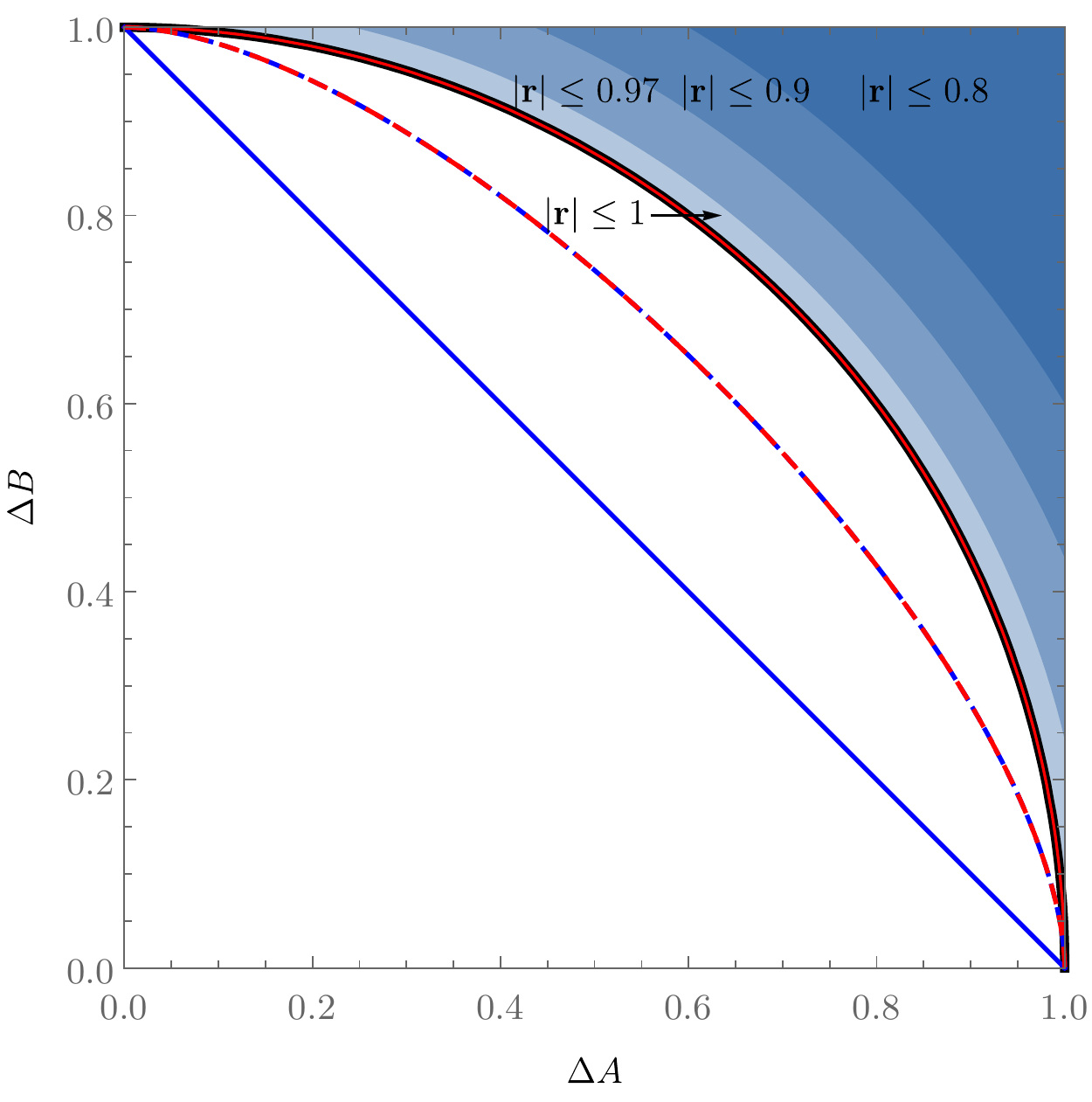} &
			 &
			\includegraphics[width=0.47\textwidth]{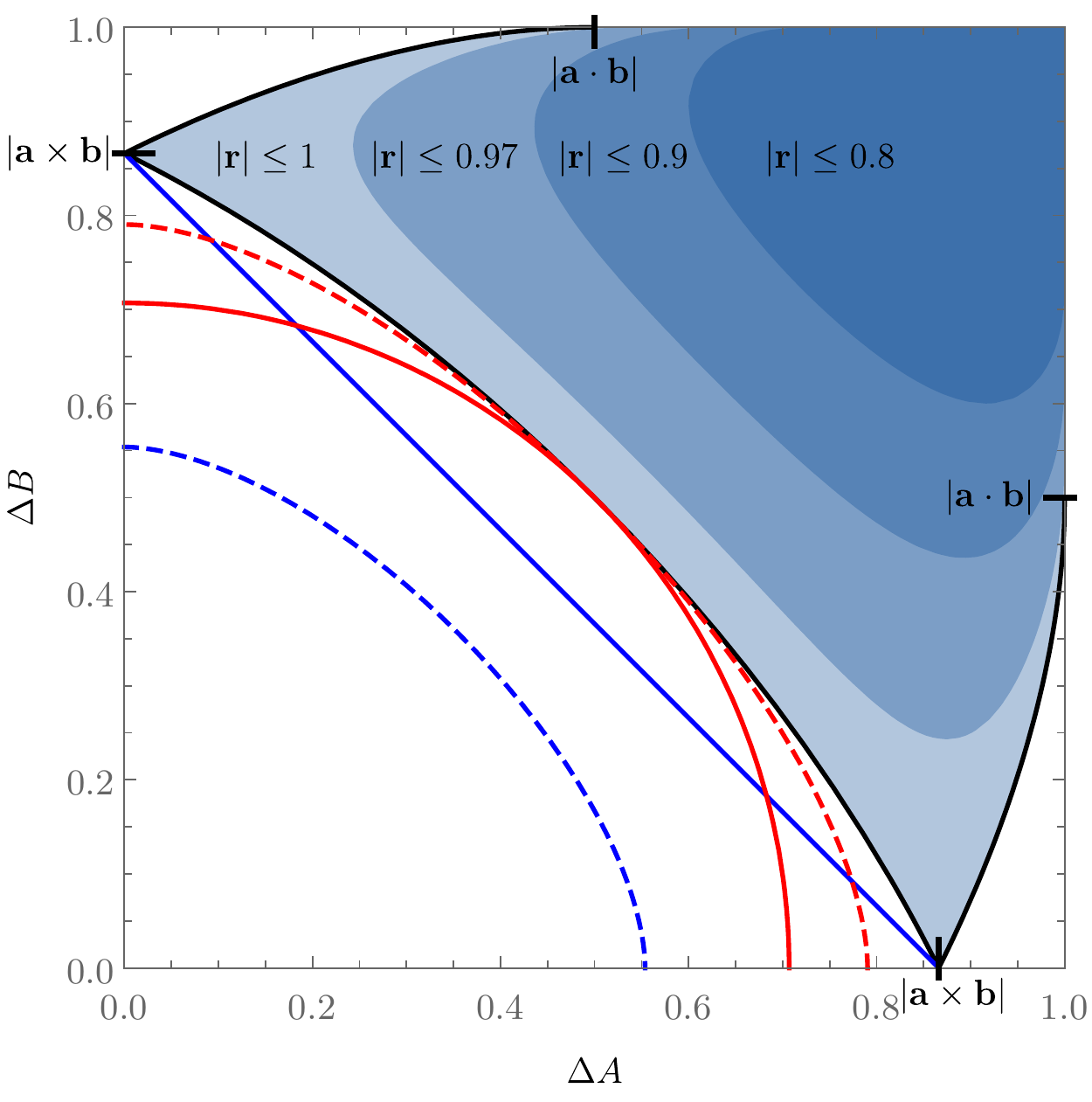}\\
			(a) & & (b)\\
		\end{tabular}
	\end{center}
	\caption{
	Representation of the uncertainty regions (filled areas) characterised by Eq.~\eqref{eqn:SIReln2obs}, shown here (a) for $\vprod{a}{b}=0$ and (b) for $|\vprod{a}{b}]=\frac{1}{2}$. The values on the thick black curves saturate the uncertainty relation~\eqref{eqn:SIReln2obs}. The darker areas represent the allowed values of $(\Delta A, \Delta B)$ for mixed states with bounded Bloch vector norms $|\vect{r}| \le 0.97, 0.9$, and $0.8$, characterised using the first inequality in Eq.~\eqref{eqn:SIReln2obs}.
	\\
	We compare here these uncertainty regions to the bounds given by Eqs.~\eqref{eqn:Busch1} and \eqref{eqn:Busch2} (blue and red curves, which touch the uncertainty region for $(\Delta A,\Delta B)=(|\vect{a}\times\vect{b}|,0)$ or $(0,|\vect{a}\times\vect{b}|)$, and for $(\Delta A,\Delta B)=(\sqrt{\frac{1 - |\vect{a}\cdot\vect{b}|}{2}},\sqrt{\frac{1 - |\vect{a}\cdot\vect{b}|}{2}})$, respectively), as well as the two entropic uncertainty relations given by Eq.~\eqref{eqn:EntrRelnMU} and in Ref.~\cite{Garrett:1990bf,Sanchez-Ruiz:1998by,Ghirardi:2003ez,Vicente:2008oq} (blue and red dashed curves, respectively), which we translate in terms of standard deviations using $H(A) = h_2\big(\frac{1{+}\sqrt{1{-}(\Delta A)^2}}{2}\big)$ (with $h_2$ the binary entropy function; see Section~\ref{sec:entropic}). Our relation is clearly stronger than all these bounds. Note that in the special case where $\vprod{a}{b}=0$, the bounds given by Eqs.~\eqref{eqn:SIReln2obs} and~\eqref{eqn:Busch2} coincide, as do the two entropic uncertainty relations.}
	\label{fig:UncertanityRelns}
\end{figure}

It is worth noting that, using the relation
\begin{equation}\label{eqn:commutatorNormLink}
\sqrt{1 - (\vect{a}\cdot\vect{b})^2} = |\vect{a}\times\vect{b}| = \big| \tfrac{1}{2i} [A,B]\big|
\end{equation}
where, in the last term, $|\cdot|$ denotes the operator norm,
Eq.~\eqref{eqn:SIReln2obs} can be expressed in terms of this commutator -- a measure of the incompatibility of $A$ and $B$ -- instead of the inner product $\vprod{a}{b}$.
Interestingly, one can thus visualise this incompatibility as the area of the parallelogram defined by $\vect{a}$ and $\vect{b}$.

Theorem~\ref{thm:SIReln2obs} provides a \emph{tight} state-independent uncertainty relation. 
One can indeed verify, as we did for Lemma~\ref{lemma_2obs}, that any pair of values $(\Delta A, \Delta B)$ satisfying Eq.~\eqref{eqn:SIReln2obs} can be attained, for instance (in the nontrivial case $|\vect{a}\cdot\vect{b}| < 1$), by the state with Bloch vector
\begin{equation}
		\vect{r} = \frac{1}{1-(\vect{a}\cdot\vect{b})^2} \Big[ \sqrt{1-(\Delta A)^2} \, \big( \vect{a} - (\vect{a}\cdot\vect{b}) \, \vect{b} \big) + \tau \sqrt{1-(\Delta B)^2} \, \big( \vect{b} - (\vect{a}\cdot\vect{b}) \, \vect{a} \big) \Big]
\end{equation}
with $\tau = \sgn (\vect{a}\cdot\vect{b})$, which is a valid Bloch vector (i.e. its norm is at most $1$) if and only if Eq.~\eqref{eqn:SIReln2obs} is satisfied. Note that it saturates the first inequality in~\eqref{eqn:SIReln2obs}.
To saturate both inequalities in~\eqref{eqn:SIReln2obs} for a given value of $\Delta A$, consider the pure states in the $\vect{a}\vect{b}$-plane with Bloch vectors
\begin{equation}
		\vect{r}_\pm = \sqrt{1-(\Delta A)^2} \ \vect{a} \pm \tau \, \frac{\Delta A}{\sqrt{1-(\vect{a}\cdot\vect{b})^2}} \big( \vect{b} - (\vect{a}\cdot\vect{b}) \, \vect{a} \big) \, .
\end{equation}
It can be checked that $\vect{r}_+$ always saturates both inequalities in~\eqref{eqn:SIReln2obs}, while $\vect{r}_-$ does so only when $\Delta A \le |\vect{a}\cdot\vect{b}|$, $\Delta A=1$ or $\vect{a}\cdot\vect{b} = 0$ (see Figure~\ref{fig:UncertanityRelns}). Note that in any case, $\Delta B(\vect{r}_+) \le \Delta B(\vect{r}_-)$ (with equality for $\Delta A=0$, $\Delta A=1$ or $\vect{a}\cdot\vect{b} = 0$).

The Bloch vectors $\vect{r_+}$ and $\vect{r_-}$ hence completely characterise the boundary of the uncertainty region in the space $\Delta A \times \Delta B$ where Eq.~\eqref{eqn:SIReln2obs} is saturated; the boundary is completed by the line segments $(\Delta A \ge |\vprod{a}{b}|, \Delta B = 1)$ and $(\Delta A = 1, \Delta B \ge |\vprod{a}{b}|)$; see Figure~\ref{fig:UncertanityRelns}.
Often, one is interested only in the \emph{monotone closure} of the uncertainty region (as in Refs.~\cite{Abdelkhalek:2015cr,Dammeier:2015zt}); that is, the closure under increasing either coordinate of the set of realisable pairs $(\Delta A,\Delta B)$.
The lower bound of this uncertainty region is obtained by the state with Bloch vector $\vect{r_+}$ for $\Delta A \le \sqrt{1-(\vect{a}\cdot\vect{b})^2}$, since for this state $\Delta B=0$ when $\Delta A=\sqrt{1-(\vect{a}\cdot\vect{b})^2}=|\vect{a}\times\vect{b}|$.

This monotone closure can also be characterised by an uncertainty relation that follows\footnote{For a fixed value of $|\vect{r}|$, one can similarly show (e.g., by replacing $(\Delta A)^2$ by $\big( 1-\frac{1-(\Delta A)^2}{|\vect{r}|^2} \big)$ and $(\Delta B)^2$ by $\big( 1-\frac{1-(\Delta B)^2}{|\vect{r}|^2} \big)$ in the calculations) that the monotone closure is given by
	\begin{equation*}
		(\Delta A)^2 + (\Delta B)^2 + 2 \, | \vect{a}\cdot\vect{b} | \, \sqrt{|\vect{r}|^2 - \big( 1-(\Delta A)^2 \big)} \, \sqrt{|\vect{r}|^2 - \big( 1-(\Delta B)^2 \big)}
		\ \ge \ 2 - \big(1 + (\vect{a}\cdot\vect{b})^2\big) \, |\vect{r}|^2 .
	\end{equation*}} from Eq.~\eqref{eqn:SIReln2obs}:
\begin{Theorem}
	The monotone closure of~\eqref{eqn:SIReln2obs} is given by
	\begin{equation}\label{eqn:SIReln2obs_monotone_closure}
		(\Delta A)^2 + (\Delta B)^2 + 2 \, | \vect{a}\cdot\vect{b} | \, \Delta A \, \Delta B
		\ \ge \ 1 - (\vect{a}\cdot\vect{b})^2 .
	\end{equation}
\end{Theorem}
\begin{proof}
	Let us first show that Eq.~\eqref{eqn:SIReln2obs_monotone_closure} is equivalent to Eq.~\eqref{eqn:SIReln2obs} if $(\Delta A)^2 + (\Delta B)^2 \le 1 - (\vect{a}\cdot\vect{b})^2$.
	In this case, then $(\Delta A)^2 + (\Delta B)^2 \ \le \ 1 + (\vect{a}\cdot\vect{b})^2$ also, and we can write Eq.~\eqref{eqn:SIReln2obs} as
	\begin{align}
		 \ 4 \, (\vect{a}\cdot\vect{b})^2 \, \left[1-(\Delta A)^2\right] \, \left[1-(\Delta B)^2\right]
		- \left[ 1 + (\vect{a}\cdot\vect{b})^2 - (\Delta A)^2 - (\Delta B)^2 \right]^2 \qquad \notag\\
		= \  4 \, (\vect{a}\cdot\vect{b})^2 \, (\Delta A)^2 \, (\Delta B)^2 
		- \big[ 1 - (\vect{a}\cdot\vect{b})^2 - (\Delta A)^2 - (\Delta B)^2 \big]^2 & \ \ge  \ 0.
	\end{align}
	Thus, still under the assumption that $(\Delta A)^2 + (\Delta B)^2 \ \le \ 1 - (\vect{a}\cdot\vect{b})^2$, we have 
	\begin{equation}
		 2 \, |\vect{a}\cdot\vect{b}| \, \Delta A \, \Delta B
		\ \ge \ 1 - (\vect{a}\cdot\vect{b})^2 - (\Delta A)^2 - (\Delta B)^2,
	\end{equation}
	which is precisely Eq.~\eqref{eqn:SIReln2obs_monotone_closure}.
	
	Equation~\eqref{eqn:SIReln2obs_monotone_closure} clearly defines its own monotone closure: if $(\Delta A,\Delta B)$ satisfy it, then so do any $(\Delta A', \Delta B')$ with $\Delta A' \ge \Delta A$ and $\Delta B' \ge \Delta B$. 
	Since, as we just showed, Eq.~\eqref{eqn:SIReln2obs} is equivalent to Eq.~\eqref{eqn:SIReln2obs_monotone_closure} in the region where $(\Delta A)^2 + (\Delta B)^2 \le 1-(\vprod{a}{b})^2$, its monotone closure in that region is given by~\eqref{eqn:SIReln2obs_monotone_closure}. 
	Furthermore, since all points with $(\Delta A)^2 + (\Delta B)^2 = 1-(\vprod{a}{b})^2$ satisfy Eq.~\eqref{eqn:SIReln2obs_monotone_closure}, and hence also (the equivalent, for these points) Eq.~\eqref{eqn:SIReln2obs}, then the whole region where $(\Delta A)^2 + (\Delta B)^2 \ge 1-(\vprod{a}{b})^2$ is in the monotone closure of~\eqref{eqn:SIReln2obs}. 
	All together, the monotone closure of~\eqref{eqn:SIReln2obs} is thus composed of the points in the region where $(\Delta A)^2 + (\Delta B)^2 \le 1-(\vprod{a}{b})^2$ which satisfy~\eqref{eqn:SIReln2obs_monotone_closure}, and of all points in the region $(\Delta A)^2 + (\Delta B)^2 \ge 1-(\vprod{a}{b})^2$; 
	since the latter clearly also satisfy~\eqref{eqn:SIReln2obs_monotone_closure}, this equation is sufficient to fully characterise the monotone closure of~\eqref{eqn:SIReln2obs}.
\end{proof}

Note that the relation given by Eq.~\eqref{eqn:SIReln2obs_monotone_closure} can readily be used to derive the weaker relations given by Eqs.~\eqref{eqn:Busch1} and~\eqref{eqn:Busch2}, which are thus in turn seen to follow from~\eqref{eqn:SIReln2obs}.
In particular, we can obtain~\eqref{eqn:Busch1} by noting that
\begin{equation}
	(\Delta A + \Delta B)^2 \ \ge \ (\Delta A)^2 + (\Delta B)^2 + 2 \, | \vect{a}\cdot\vect{b} | \, \Delta A \, \Delta B
			\ \ge \ 1 - (\vect{a}\cdot\vect{b})^2
\end{equation}
and thus
\begin{equation}
	\Delta A + \Delta B \ \ge \ \sqrt{1 - (\vect{a}\cdot\vect{b})^2} = |\vect{a}\times\vect{b}|.
\end{equation}
Eq.~\eqref{eqn:Busch2} can be obtained by seeing that
\begin{equation}
	(1 + |\vect{a}\cdot\vect{b}|) \, \big[ (\Delta A)^2 + (\Delta B)^2 \big] \ \ge \ (\Delta A)^2 + (\Delta B)^2 + 2 \, | \vect{a}\cdot\vect{b} | \, \Delta A \, \Delta B
			\ \ge \ 1 - (\vect{a}\cdot\vect{b})^2
\end{equation}
and thus
\begin{equation}
(\Delta A)^2 + (\Delta B)^2 \ge 1 - |\vect{a}\cdot\vect{b}|.
\end{equation}

Figure~\ref{fig:UncertanityRelns} also shows these two bounds in comparison to the tight relation~\eqref{eqn:SIReln2obs} that we derived.
Note that if $A$ and $B$ are orthogonal spin measurements (i.e., $\vect{a}\cdot\vect{b} = 0$) then Eq.~\eqref{eqn:Busch2} reduces to Eq.~\eqref{eqn:SIReln2obs} and is tight, but for all other cases neither~\eqref{eqn:Busch1} nor~\eqref{eqn:Busch2} are tight: these bounds are only obtained by states with $(\Delta A,\Delta B)=(|\vect{a}\times\vect{b}|,0)$ or $(0,|\vect{a}\times\vect{b}|)$ for Eq.~\eqref{eqn:Busch1}, and $\Delta A = \Delta B = \sqrt{\frac{1 - |\vect{a}\cdot\vect{b}|}{2}}$ for Eq.~\eqref{eqn:Busch2}.

We note finally that it is also possible to express the relation~\eqref{eqn:SIReln2obs_monotone_closure} in a further form that may be useful in understanding the lower bound of the uncertainty region.
Specifically, Eq.~\eqref{eqn:SIReln2obs_monotone_closure} is satisfied if and only if\footnote{While finishing this manuscript we became aware that the second part of Eq.~\eqref{eqn:SIReln2obs_monotone_closure2} had been independently derived by P. Busch using a geometric argument~\cite{Busch:2015zh}.}
\begin{equation}\label{eqn:SIReln2obs_monotone_closure2}
	(\Delta A)^2 + (\Delta B)^2 \ \ge \ 1 \quad \text{or} \quad \Delta A \, \sqrt{1-(\Delta B)^2} + \Delta B \, \sqrt{1-(\Delta A)^2} \ \ge \ \sqrt{1 - (\vect{a}\cdot\vect{b})^2} \, ,
\end{equation}
where the lower bound of the uncertainty region lies, as we saw before, in the region $(\Delta A)^2 + (\Delta B)^2 \le 1$ and is expressed by the second half of Eq.~\eqref{eqn:SIReln2obs_monotone_closure2}.
The proof that this alternative form is equivalent to Eq.~\eqref{eqn:SIReln2obs_monotone_closure} mirrors one given in Ref.~\cite{Branciard:2014mz} so we do not give it here, but it is important to note that the second half of this relation is violated by some allowable uncertainty pairs (e.g. $\Delta A=\Delta B=1$), and hence the first disjunction is essential for the relation to be valid for all pairs $(\Delta A,\Delta B)$.

Using Eq.~\eqref{eqn:commutatorNormLink} we note that the right-hand side of the second inequality in Eq.~\eqref{eqn:SIReln2obs_monotone_closure2} is precisely $|{\tfrac{1}{2i}}[A,B]|$, so the monotone closure of the uncertainty region can be expressed as a relation \emph{only} on the uncertainties bounded by a function of the commutator of $A$ and $B$.

\subsection{Uncertainty relations for $n$ Pauli observables}

Similarly to the previous case of two measurements, one can now use Lemma~\ref{lemma_nobs} to derive state-independent uncertainty relations for more observables. Namely, using $\braket{A_i} = \pm \sqrt{1-(\Delta A_i)^2}$, we obtain:

\begin{Theorem}\label{thm:SIReln_nobs}
For any $n$ Pauli observables $A_1,\dots,A_n$ with Bloch vectors $\vect{a}_1,\dots,\vect{a}_n$ spanning the whole Bloch sphere, every quantum state $\rho = \frac12 (\vect{1}+\vect{r}\cdot \bm{\sigma})$ satisfies the relation
\begin{eqnarray}
	\exists \, \tau_1, \ldots, \tau_n = \pm 1 , \quad 
	\sum_{1 \le i \le n} \!\! m_{i,i} \, (\Delta A_i)^2 - \sum_{1 \le i \neq j \le n} \!\!\! \tau_i \, \tau_j \, m_{i,j} \sqrt{1-(\Delta A_i)^2} \sqrt{1-(\Delta A_j)^2} \hspace{15mm} \nonumber \\[-2mm]
	= \, \Big( \!\sum_{1 \le i \le n} \!\! m_{i,i} \Big)-|\vect{r}|^2 \,\ge\, \Big( \!\sum_{1 \le i \le n} \!\! m_{i,i} \Big) -1, \label{eqn:SIReln_nobs}
\end{eqnarray}
where the $m_{i,j}$ coefficients are the elements of the symmetric matrix $(M^+)^\T M^+$, with $M^+$ the pseudoinverse of the matrix $M=(\vect{a}_1^\T,\dots,\vect{a}_n^\T)^\T$.

The relation~\eqref{eqn:SIReln_nobs} can furthermore be written without the existential quantifier as an inequality ($\ge$) instead of an equality if one replaces all $-\tau_i \, \tau_j \, m_{i,j}$ by $+|m_{i,j}|$. 
This inequality then also holds if the Bloch vectors do not span the whole Bloch sphere, and is saturated if and only if $\vect{r} \in \Span \{\vect{a}_1,\dots,\vect{a}_n\}$ and $m_{i,j} \braket{A_i} \braket{A_j} \le 0$ for all $i \neq j$.
\end{Theorem}

Note that the quantifiers ``$\exists \, \tau_i$'' have been introduced in Eq.~\eqref{eqn:SIReln_nobs} to make the relation state-independent (i.e.,\ so that it is a constraint that can be evaluated solely on the (measurable) $\Delta A_i$'s, where no other terms depend on the quantum state). 
In practice, however, if $\rho$ is known, one can simply take the signs to be $\tau_i = \sgn \braket{A_i}$.

As for Lemma~\ref{lemma_nobs}, Theorem~\ref{thm:SIReln_nobs} does not give tight relations for $n>3$.
However, it can easily be made tight by requiring, in addition, that $(\tau_1\sqrt{1-(\Delta A_1)^2},\dots,\tau_n\sqrt{1-(\Delta A_n)^2})\coloneqq\vect{u}$ satisfies $MM^+\vect{u}=\vect{u}$.
Crucially, note that this further condition is state-independent, just like Eq.~\eqref{eqn:SIReln_nobs}.

Let us illustrate again this relation for some specific examples. For the case of $n = 3$ Pauli observables $A$, $B$ and $C$ with linearly independent Bloch vectors, we obtain the tight relation\footnote{We recover, in particular, a relation given in Ref.~\cite{Li:2015wf} for the specific case $\vect{a} \cdot \vect{c} = \vect{b} \cdot \vect{c} = 0$ (although the relation in~\cite{Li:2015wf} omitted the quantifier ``$\exists\, \tau_{ab} = \pm 1$'').}
\begin{eqnarray}
	&& \hspace{-7mm} \forall \, \rho, \ \exists \, \tau_A, \tau_B, \tau_C = \pm 1 , \nonumber\\[2mm]
	&& \left| \vect{b}\times\vect{c} \right|^2 \, (\Delta A)^2 + \left| \vect{a}\times\vect{c} \right|^2 \, (\Delta B)^2 + \left| \vect{a}\times\vect{b} \right|^2 \, (\Delta C)^2 \nonumber \\
	&& \quad - \, 2 \, \tau_A\, \tau_B\, (\vect{b}\times\vect{c})\cdot(\vect{c}\times\vect{a}) \sqrt{1-(\Delta A)^2} \sqrt{1-(\Delta B)^2} \nonumber \\
	&& \quad \quad - \, 2 \, \tau_A\, \tau_C\, (\vect{b}\times\vect{c})\cdot(\vect{a}\times\vect{b}) \sqrt{1-(\Delta A)^2} \sqrt{1-(\Delta C)^2} \nonumber \\
	&& \quad \quad \quad - \, 2 \, \tau_B\, \tau_C\, (\vect{c}\times\vect{a})\cdot(\vect{a}\times\vect{b}) \sqrt{1-(\Delta B)^2} \sqrt{1-(\Delta C)^2} \nonumber \\
	&& \quad \quad \quad \quad = \ |\vect{a}\times\vect{b}|^2 + |\vect{a}\times\vect{c}|^2+|\vect{b}\times\vect{c}|^2 - V^2 \, |\vect{r}|^2 \ \ge \ 2 - 2 \,(\vect{a} \cdot \vect{b})(\vect{a} \cdot \vect{c})(\vect{b} \cdot \vect{c}) \quad \quad\label{eqn:3obsStdDevs}
\end{eqnarray}
where 
\begin{equation}
V^2 = (\,\vect{a}\cdot(\vect{b}\times\vect{c})\,)^2 
	= |\vect{a}\times\vect{b}|^2 + |\vect{a}\times\vect{c}|^2+|\vect{b}\times\vect{c}|^2 + 2(\vprod{a}{b})(\vprod{a}{c})(\vprod{b}{c})-2
	> 0.
\end{equation}
It is interesting to note that $V$ corresponds to the volume of the parallelepiped defined by the Bloch vectors $\vect{a}$, $\vect{b}$ and $\vect{c}$, and hence is a measure of the mutual incompatibility of the observables.
When the three measurements are orthogonal ($\vect{a} \cdot \vect{b} = \vect{a} \cdot \vect{c} = \vect{b} \cdot \vect{c} = 0$, $V^2=1$), then we get~\cite{Hofmann:2003xd,Li:2015wf,Dammeier:2015zt}
\begin{equation}
	(\Delta A_1)^2 + (\Delta A_2)^2 + (\Delta A_3)^2 \ = \ 3 - |\vect{r}|^2 \ \ge \ 2 \,.
\end{equation}

In the case of the regular tetrahedron measurements described earlier, we find 
\begin{eqnarray}
	&& \hspace{-8mm} \forall \, \rho, \ \exists \, \tau_1, \ldots, \tau_4 = \pm 1, \nonumber \\[2mm]
	&& 3 \sum_{1 \le i \le 4} (\Delta A_i)^2 + \sum_{1 \le i \neq j \le 4} \tau_i \, \tau_j \, \sqrt{1-(\Delta A_i)^2} \sqrt{1-(\Delta A_j)^2} \ = \ 12 - \frac{16}{3} \, |\vect{r}|^2 \ \ge \ \frac{20}{3} \, .
\end{eqnarray}

\section{Entropic uncertainty relations}
\label{sec:entropic}

In the case of qubits, the Shannon entropy of a Pauli observable $A$ as defined in Eq.~\eqref{eqn:def_HA} can be directly expressed in terms of the expectation value $\braket{A}$, namely:
\begin{equation}\label{eqn:binaryEntropy}
	H(A) = h_2 \Big(\frac{1+\braket{A}}{2} \Big) = h_2 \Big(\frac{1-\braket{A}}{2} \Big) \,,
\end{equation}
where $h_2$ is the binary entropy function defined as $h_2(p) = -p \log p - (1-p) \log (1-p)$.
Denoting by $h_2^{-1}$ the inverse function of $h_2$ restricted to the domain $p \in [0,\frac12]$, one can invert this relation and obtain
\begin{equation}
	\braket{A} = \pm f\big( H(A) \big) \quad \textrm{with} \quad f\big( H(A) \big) = 1 - 2 \, h_2^{-1} \big( H(A) \big) = |\braket{A}| \, \ge \, 0\, .
\end{equation}
This allows us now to express the relations of Lemmas~\ref{lemma_2obs} and~\ref{lemma_nobs} in terms of the Shannon entropies $H(A_i)$. 
Note that because of the $\pm$ sign above, the same care as before must be taken when transforming these relations into entropic ones. 
In fact, the exact same analysis as in the previous section (which we shall not repeat explicitly) can be carried out: one can simply replace all terms $\sqrt{1-(\Delta A)^2}$ by $f\big( H(A) \big)$ and the variances $(\Delta A)^2$ by $1 - f\big( H(A) \big)^2$.

For instance, one obtains the following tight relation for two Pauli observables (similarly to Theorem~\ref{thm:SIReln2obs}, with the same conditions for saturation):
\begin{equation}\label{eqn:SIReln2obs_entropy}
		f\big( H(A) \big)^2 + f\big( H(B) \big)^2 - 2 \, |\vect{a}\cdot\vect{b}|\, f\big( H(A) \big)\, f\big( H(B) \big)
		\ \le \  \big(1-(\vect{a}\cdot\vect{b})^2 \big) \, |\vect{r}|^2
		\ \le \ 1 - (\vect{a}\cdot\vect{b})^2.
	\end{equation}

\begin{figure}[t]
	\begin{center}
		\begin{tabular}{ccc}
			\includegraphics[width=0.47\textwidth]{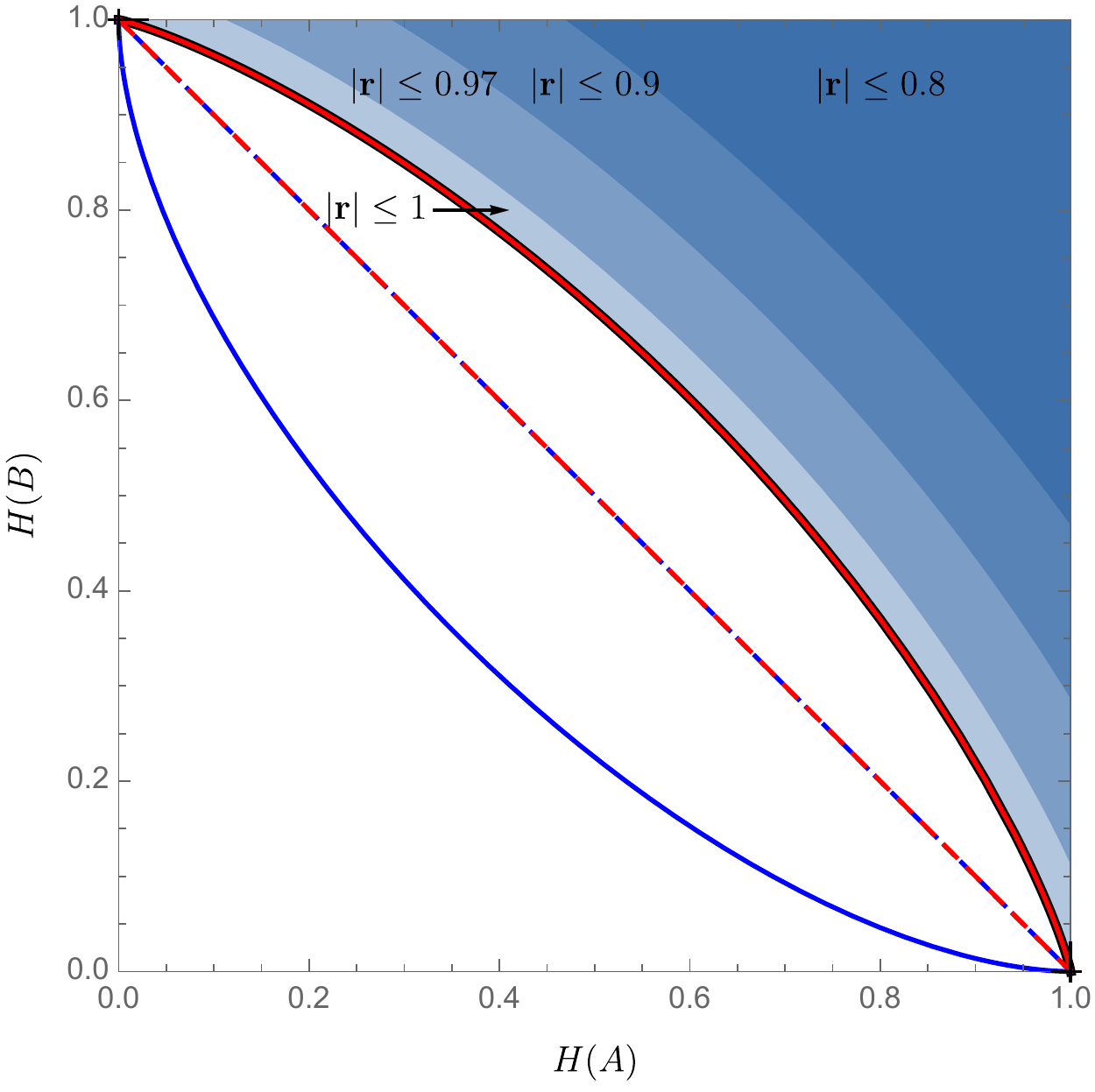} &
			 &
			\includegraphics[width=0.47\textwidth]{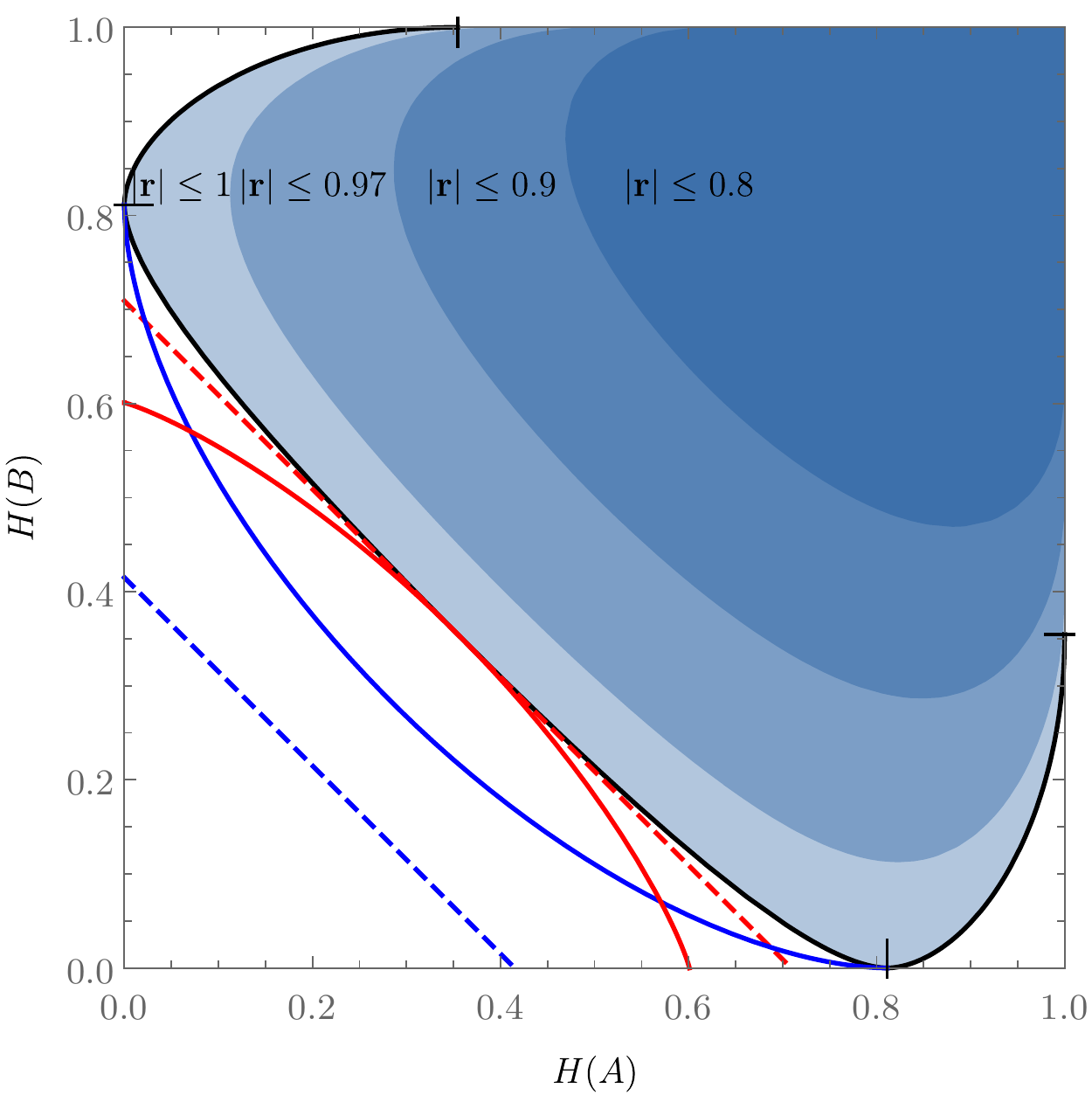}\\
			(a) & & (b)\\
		\end{tabular}
	\end{center}
	\caption{
	Analogous figure to Fig.~\ref{fig:UncertanityRelns}, with the uncertainty regions and relations now shown in the Shannon entropy domain, for (a) $\vprod{a}{b}=0$ and (b) $|\vprod{a}{b}|=\frac{1}{2}$, i.e., for maximum overlaps $c = \sqrt{\frac{1+|\vect{a}\cdot\vect{b}|}{2}} = \frac{1}{\sqrt{2}}$ and $c=\frac{\sqrt{3}}{2}$, respectively. The filled area corresponds to relation~\eqref{eqn:SIReln2obs_entropy} (which is saturated for the values of $(H(A), H(B))$ on the thick black curve), with the darker areas representing values attainable for mixed states with bounded Bloch vector norms.
	The dashed blue and red curves correspond to the Maassen--Uffink bound~\eqref{eqn:EntrRelnMU} and to that of Refs.~\cite{Garrett:1990bf,Sanchez-Ruiz:1998by,Ghirardi:2003ez,Vicente:2008oq}, respectively. In (b), the former is clearly suboptimal (it coincides with the latter only for $\vprod{a}{b}=0$), while the latter is the optimal bound that can be put on the sum $H(A) + H(B)$, and is thus tangential to the uncertainty region characterised by~\eqref{eqn:SIReln2obs_entropy} (it touches it at the point $H(A)=H(B)=h_2\big(\frac{1+|\vprod{a}{b}|}{2}\big)$ for $|\vprod{a}{b}| \gtrsim 0.391$, and on two symmetric points when $|\vprod{a}{b}| \lesssim 0.391$; the critical value $|\vprod{a}{b}| \simeq 0.391$ corresponds to $\arccos|\vprod{a}{b}| \simeq 0.585$~\cite{Sanchez-Ruiz:1998by}, or $c \simeq 0.834$~\cite{Vicente:2008oq}).
	The solid blue and red curves correspond to the bounds~\eqref{eqn:Busch1} and~\eqref{eqn:Busch2}, respectively, translated in terms of entropies using $\Delta A = \sqrt{1 - f\big( H(A) \big)^2}$ (see text).}
	\label{fig:EntropicRelns}
\end{figure}

Figure~\ref{fig:EntropicRelns} shows the bound given by Eq.~\eqref{eqn:SIReln2obs_entropy} in the space of Shannon entropies $H(A),H(B)$ for $\vprod{a}{b}=0$ and $|\vprod{a}{b}|=\frac{1}{2}$, in comparison to the weaker Massen--Uffink bound~\eqref{eqn:EntrRelnMU} -- with the maximum overlap given here by $c = \sqrt{\frac{1+|\vect{a}\cdot\vect{b}|}{2}}$ -- as well as the better, but still not tight, bound given in Refs.~\cite{Garrett:1990bf,Sanchez-Ruiz:1998by,Ghirardi:2003ez,Vicente:2008oq}.

\section{Higher-dimensional systems}

One may wonder whether our approach can be generalised to give state-independent uncertainty relations for qutrits or higher-dimensional systems. As we will see, while it remains possible to obtain (not necessarily tight) relations on the expectation values of a set of observables, the peculiar geometry of the generalised Bloch sphere for qutrits (and beyond), as well as the increased number of eigenvalues of the observables, prevents us from expressing these relations in terms of standard deviations or entropies, as we could do for qubits.

In a $d$-dimensional Hilbert space an arbitrary traceless observable $A$ can be expressed in terms of a set of $d\times d$ traceless Hermitian matrices $\lambda_i$, $1\le i \le d^2-1$, which generate the group $\mathrm{SU(}d\mathrm{)}$.
These operators satisfy the commutation relations
\begin{equation}
	[\lambda_i,\lambda_j]=2\,i\,f_{ijk}\,\lambda_k,\ \ \ \ \{\lambda_i,\lambda_j\}=\frac{4}{d}\,\delta_{ij} \vect{1} + 2\,d_{ijk}\,\lambda_k \, ,
\end{equation}
where $[\cdot,\cdot]$ and $\{\cdot,\cdot\}$ are the commutator and anticommutator, $\delta_{ij}$ is the Kronecker delta, $f_{ijk}$ and $d_{ijk}$ are antisymmetric and symmetric structure constants of $\mathrm{SU(}d\mathrm{)}$, respectively, and where the summation over repeated indices is implicit.
For qutrits the operators $\lambda_1$ to $\lambda_8$ are the Gell-Mann matrices, for example.
As for the two-dimensional case, taking $\vect{\bm{\lambda}}=(\lambda_1,\dots,\lambda_{d^2-1})^\T$ we can write any traceless observable $A$ as $A=\vprod{a}{\bm{\lambda}}$ and thus represent it by its generalised Bloch vector $\vect{a}$.

An arbitrary state $\rho$ can similarly be written in terms of its generalised Bloch vector $\vect{r}$ as 
\begin{equation}
\rho= \frac{1}{d}\vect{1} + \frac{1}{2}\vprod{r}{\bm{\lambda}} \, ,
\end{equation}
where now $|\vect{r}| \le \sqrt{2(1-\frac{1}{d})}$ (with equality for pure states).
However, for $d \geq 3$ it is \emph{not} the case that \emph{any} vector $\vect{r}$ with $|\vect{r}|\le \sqrt{2(1-\frac{1}{d})}$ represents a valid quantum state~\cite{Bertlmann:2008aa,Kimura:2003aa}: the set of valid quantum states (i.e., the Bloch vector space) is a strict subset of the unit sphere in $d$ dimensions.

The expectation value of $A$ in the state $\rho$ can still be expressed as $\braket{A}=\vect{a}\cdot \vect{r}$, so that Lemma~\ref{lemma:generalForm}, with $M$ and $\vect{u}$ defined as before, remains valid for higher-dimensional systems. This allows one to derive state-independent relations for the expectation values $(\braket{A_1}, \ldots, \braket{A_n})$, as we did for qubits in Section~\ref{sec:uncertainties}. Note that these relations may not be tight, as the vectors $\vect{r}$ saturating~\eqref{eqn:PseudoInvSD} may not correspond to valid quantum states.%
\footnote{
It is also worth noting that the ``ellipsoid condition'' of Kaniewski \emph{et al.}~\cite{Kaniewski:2014aa}, which is a higher-dimensional version of Lemma~\ref{lemma:generalForm} for \emph{binary-valued} measurements, is state-dependent for $d\ge 3$, so it also fails to give the type of state-independent uncertainty relation we would like in this scenario.
}

Contrary to the case of qubits however, for $d \ge 3$ one cannot directly translate these relations to express them solely in terms of standard deviations or entropies; indeed, because of the larger number of eigenvalues and of the geometry of the generalised Bloch sphere, the expectation value $\braket{A}$ does not contain all the information about the uncertainty of $A$. For instance, the relation between $(\Delta A)^2$, $\vect{a}$ and $\vect{r}$ is not simply given by Eq.~\eqref{eqn:StdDev} (where $\braket{A}=\vect{a}\cdot \vect{r}$), but by
\begin{equation}\label{eqn:StdDevHighDim}
	(\Delta A)^2 = \frac{2}{d}|\vect{a}|^2 + \vprod{a'}{r} - (\vprod{a}{r})^2,
\end{equation}
where $\vect{a'}$ is a $d$-dimensional vector with components $a'_k=a_i \, a_j \, d_{ijk}$~\cite{Li:2015wf}.
Thus, the uncertainty of $A$ for a state $\rho$ no longer depends only on the angle between the Bloch vectors $\vect{r}$ and $\vect{a}$.
Furthermore, in contrast to qubit operators, for three-or-more-level systems there are pairs of non-commuting observables that have a common eigenstate and hence can simultaneously have zero variance.
In general, there is no simple analytical description of the Bloch space for qudits~\cite{Bertlmann:2008aa,Kimura:2003aa}, so it seems implausible to give generalised forms of tight, state-independent uncertainty relations for such systems.

For certain choices of $A$ and $B$ a complete analysis of the set of obtainable values for $\Delta A$ and $\Delta B$ is nevertheless tractable (at least for qutrits), and it is possible to give tight state-independent uncertainty relations.
Similarly, for well chosen $A$ and $B$ a more general higher-dimensional analysis is possible if one is prepared to settle for relations that are not tight.
In Ref.~\cite{Dammeier:2015zt}, for example, such behaviour is analysed for angular momentum observables in orthogonal directions.
Such an approach, however, lacks the generality of the approach possible for qubits, and is necessarily, at least in part, \emph{ad hoc}.

\section{Discussion} 
\label{sec:discussion}

By exploiting the relationship between the expectation values of Pauli observables and standard measures of uncertainty (such as the standard deviation), we have derived tight state-independent uncertainty relations for Pauli measurements on qubits.
These uncertainty relations completely characterise the allowed values of uncertainties for such observables.
Furthermore, we give the bounds on all these relations in terms of the norm $|\vect{r}|$ of the Bloch vector representing the state $\rho=\frac12(\vect{1}+\vprod{r}{\bm{\sigma}})$ -- which is directly linked to the purity of the state -- so that, if a bound on this is known, tighter (partially state-dependent) uncertainty relations can be obtained; for pure states $|\vect{r}|=1$ and the most general form is recovered.
The approach we take is general and, although we explicitly give tight uncertainty relations for arbitrary pairs and triples of Pauli observables, it can be used to give tight uncertainty relations for sets of arbitrarily many observables.

While we have focused on giving these uncertainty relations in terms of the standard deviations and variances of the observables, we showed how these can easily be rewritten in terms of Shannon entropies to give tight entropic uncertainty relations, and did so explicitly for pairs of observables.
These relations can furthermore be translated into uncertainty relations for any measure of uncertainty that depends only on the 
probability distribution, $\{\frac12(1+\vprod{a}{r}),\ \frac12(1-\vprod{a}{r})\}$, of an observable $A=\vect{a}\cdot\bm{\sigma}$ for a state $\rho=\frac12 (\vect{1}+\vect{r}\cdot \bm{\sigma})$, such as R\'enyi entropies.
Indeed, one may reasonably argue that the product $\vprod{a}{r}=\braket{A}$ is the \emph{only} parameter that an uncertainty measure for $A$ can depend on, and thus that our approach covers all kinds of preparation uncertainty relations for qubits.

Although we have given explicit uncertainty relations only for Pauli observables, it is simple to extend them to arbitrary qubit measurements.
To do so, note that one can write any observable $A$ in a two-dimensional Hilbert space as $A=\alpha\vect{1}+\vprod{a}{\bm{\sigma}}$, with $\alpha \in \mathbb{R}$ and $\vect{a}\in\mathbb{R}^3$. 
Assuming $|\vect{a}|>0$ (as otherwise $A$ is simply proportional to the identity operator and one trivially has $\Delta A=0$ for all states $\rho$), the observable $\tilde{A}=\vprod{\tilde{a}}{\bm{\sigma}}$ with $\vect{\tilde{a}} = \vect{a}/|\vect{a}|$ is a Pauli observable, and we have $\braket{\tilde{A}} = \frac{\braket{A}-\alpha}{|\vect{a}|}$ and $\Delta \tilde{A}=\frac{\Delta A}{|\vect{a}|}$. 
One can thus give an uncertainty relation involving $A$ by writing the corresponding relation that we derived for $\tilde{A}$, and then replacing $\vect{\tilde{a}}$ by $\vect{a}/|\vect{a}|$ and $\braket{\tilde{A}}$ or $\Delta \tilde{A}$ by the appropriate expression given above;
one can proceed similarly for the other observables in question.

Finally we note that, although we have not done so here, it is also possible to go beyond projective measurements and give similar relations for positive-operator valued measures (POVMs) for qubits with binary outcomes.
The two elements of any such POVM can be written in the form $A_\pm = \frac{1}{2}\big[\vect{1}\pm(\alpha\vect{1} + \vprod{a}{\bm{\sigma})}\big]$, where $|\alpha| + |\vect{a}| \le 1$. 
Attaching, for simplicity, the output values $\pm 1$ to the two POVM outcomes (this can easily be generalised), we can then define an ``effective operator'' $A = A_+ - A_-$ such that the expectation value of the POVM outcomes can be simply written as $\braket{A} = \alpha+\vprod{a}{r}$.
A similar construction to the one that led to Lemma~\ref{lemma:generalForm} can thus be used, this time utilising the ``effective operators'' and defining $\vect{u} = (\braket{A} - \alpha, \ldots)^\T$.
Lemmas~\ref{lemma_2obs} and~\ref{lemma_nobs} then still hold, after suitably substituting $\braket{A}$ for $\braket{A} - \alpha$ (and similarly for the other observables in the relation). 
To translate these into standard deviations, one can then use once more\footnote{Note that for a non-projective measurement, the standard deviation should \emph{not} be calculated as $\sqrt{\braket{A^2}-\braket{A}^2}$ using the ``effective operator'' $A$ we introduced (as in general, $A^2 \neq \vect{1}$, despite the POVM outcomes being $\pm 1$).} (for our choice of outcomes, $\pm 1$) $\Delta A=\sqrt{1-\braket{A}^2}$, where $\Delta A$ here is (in a slight abuse of notation) the POVM standard deviation, not that of the ``effective operator'', which thus leads to state-independent uncertainty relations for POVMs. 
Entropic uncertainty relations can similarly be obtained in this fashion.


\begin{acknowledgments}
AA and CB acknowledge financial support from the `Retour Post-Doctorants' program (ANR-13-PDOC-0026) of the French National Research Agency; CB also acknowledges the support of a Marie Curie International Incoming Fellowship (PIIF-GA-2013-623456) from the European Commission.
MH was supported by the ARC Centre of Excellence CE110001027.
\end{acknowledgments}


%

\end{document}